\documentclass{article}

\usepackage{amsmath,amssymb,amsfonts,amsthm,latexsym,eucal}
\usepackage{fullpage}
\usepackage[colorlinks=true]{hyperref}

\newcommand{\remove}[1]{}
\newtheorem{definition}{Definition}
\newtheorem{theorem}{Theorem}
\newtheorem{lemma}[theorem]{Lemma}
\newtheorem{proposition}[theorem]{Proposition}
\newtheorem{corollary}[theorem]{Corollary}

\newcommand{\mc}{\mathcal}
\newcommand{\fail}{\textsf{fail}}
\newcommand{\exclude}{\textsf{miss}}

\newcommand{\Enc}{\textrm{Enc}}
\newcommand{\Dec}{\textrm{Dec}}

\title{Efficient Steganography with Provable Security Guarantees}
\author{
Aggelos Kiayias\thanks{Supported by NSF CAREER grant CCR-0447808.}\\
\textsf{akiayias@cse.uconn.edu}
\and Yona Raekow \\
\textsf{yona@cse.uconn.edu} 
\and Alexander Russell\thanks{Supported by NSF CAREER grant
  CCR-0093065, and NSF grants CCR-0121277, CCR-0220264, CCR-0311368,
  and EIA-0218443.}\\
  \textsf{acr@cse.uconn.edu} \and
  Narasimha Shashidhar\\
\textsf{karpoor@cse.uconn.edu}
\vspace{4mm}\\
Department of Computer Science and Engineering\\
University of Connecticut, Storrs, CT}

\date{}
\begin{document}

\maketitle

\begin{abstract}

We provide a new provably-secure steganographic encryption protocol
that is proven secure in the complexity-theoretic
framework of Hopper et al. 

The fundamental building block of our steganographic encryption
protocol is a ``one-time stegosystem'' that allows two parties to
transmit messages of length shorter than the
shared key with \emph{information-theoretic} security guarantees.
The employment of a pseudorandom generator (PRG) permits secure transmission
of longer messages in the same way that such a generator allows
the use of one-time pad encryption for messages longer than the key in symmetric
encryption.
The advantage of our construction, compared to that of Hopper et al.,
is that it avoids the use of a pseudorandom function family and
instead relies (directly) on a pseudorandom generator in a way that provides linear improvement in the number of applications of the underlying one-way permutation per transmitted bit.  This
advantageous trade-off is achieved by substituting the pseudorandom
function family employed in the previous construction with an
appropriate combinatorial construction that has been used
extensively in derandomization, namely almost $t$-wise independent
function families.
\end{abstract}


\section{Introduction}

In a canonical steganographic scenario, Alice and Bob wish to communicate securely in the
presence of an adversary, called the ``Warden,'' who monitors whether
they exchange ``conspicuous'' messages. In particular, Alice and Bob
may exchange messages that adhere to a certain channel distributions
that represents ``inconspicuous'' communication. By controlling the
messages that are transmitted over such a channel, Alice and Bob may exchange
messages that cannot be detected by the Warden.  
There have been two approaches in formalizing this problem, one based
on information theory~\cite{DBLP:conf/ih/Cachin98,DBLP:conf/ih/ZollnerFKPPWWW98,DBLP:conf/ih/Mittelholzer99}
and one based on complexity theory~\cite{DBLP:conf/crypto/HopperLA02}.
The latter approach is more concrete %
and has the potential of allowing more efficient constructions.
Most steganographic constructions supported by provable security
guarantees are instantiations of the following basic procedure (often referred to as ``rejection-sampling'').

The problem specifies a family of message distributions (the ``channel distributions'') that provide a number of possible options for a
so-called ``covertext'' to be transmitted. Additionally, the sender and the receiver
possess some sort of private information (typically a keyed hash
function, MAC, or other similar function) that maps channel messages to a single bit.  In order to send a message bit $m$, the sender draws a covertext from the channel distribution, applies the function to the
covertext and checks whether it happens to produce the ``stegotext'' $m$ he
originally wished to transmit.  If this is the case, the covertext is
transmitted. In case of failure, this procedure is repeated.
While this is a fairly concrete procedure, there are a number of
choices to be made with both practical and theoretical
significance. From the security viewpoint, one is primarily interested in
the choice of the function that is shared between the sender and the
receiver. From a practical viewpoint, one is primarily interested in how
the channel is implemented and whether it conforms to the various
constraints that are imposed on it by the steganographic protocol
specifications (e.g., are independent draws from the channel allowed?
does the channel remember previous draws? etc.).

As mentioned above, the security of a stegosystem can be naturally
phrased in information-theoretic terms
(cf.~\cite{DBLP:conf/ih/Cachin98}) or in complexity-theoretic
terms~\cite{DBLP:conf/crypto/HopperLA02}. Informally, the latter
approach considers the following experiment for the warden-adversary:
The adversary selects a message to be embedded and
receives either covertexts that embed the message or  covertexts
simply drawn from the channel distribution (without any embedding). The
adversary is then asked to distinguish between the two cases. Clearly,
if the probability of success is very close to $1/2$ it is natural to
claim that the stegosystem provides security against such
(eavesdropping) adversarial activity. Formulation of stronger attacks (such as active attacks) is also possible.
Given the above framework, Hopper et
al.~\cite{DBLP:conf/crypto/HopperLA02} provided a provably secure
stegosystem that pairs rejection sampling with a pseudorandom function
family. Given that rejection sampling, when implemented properly and
paired with a truly random function, is indistinguishable from the
channel distribution, the security of their construction followed from
the pseudorandom function family assumption. From the efficiency
viewpoint, this construction required about 2 evaluations of the
pseudorandom function per bit transmission.  Constructing efficient
pseudorandom functions is possible either
generically~\cite{DBLP:journals/jacm/GoldreichGM86} or, more
efficiently, based on specific number-theoretic
assumptions~\cite{DBLP:journals/jacm/NaorR04}.  Nevertheless,
pseudorandom function families are a conceptually complex and fairly
expensive cryptographic primitive. For example, the evaluation of the
Naor-Reingold pseudorandom function on an input $x$ requires
$O(|x|)$ modular exponentiations.
Similarly, the generic
construction~\cite{DBLP:journals/jacm/GoldreichGM86} requires
$O(k)$ PRG doublings of the input string where $k$ is the length
of the key.

In this article we take an alternative approach to the design of
provably secure stegosystems. Our main contribution is the design of a
building block that we call a \emph{one-time stegosystem}: this is a
steganographic protocol that is meant to be used for a single message
transmission and is proven secure in an information-theoretic sense,
provided that the key that is shared between the sender and the
receiver is of sufficient length (this length analysis is part of our
result). In particular we show that we can securely transmit an $n$
bit message with a key of length $O(n + \log |\Sigma|)$; here 
$\Sigma$ is the size of the channel alphabet
(see Section~\ref{sec:parameters} for more details regarding the exact
complexity).
Our basic building block is a natural analogue of a one time-pad for
steganography. It is based on the rejection sampling technique
outlined above in combination with an explicit almost $t$-wise independent~\cite{DBLP:journals/rsa/AlonGHP92} family of functions.
We note that such combinatorial constructions have been extremely
useful for derandomization methods and here, to the best of our
knowledge, are employed for the first time
in the design of steganographic protocols.
Given a one-time stegosystem, it is fairly straightforward to
construct provably secure steganographic encryption for longer
messages by using a pseudorandom generator (PRG) to stretch a
random seed that is shared by the sender and the receiver to
sufficient length.

The resulting stegosystem is provably secure in the
computational sense of Hopper et
al.~\cite{DBLP:conf/crypto/HopperLA02} and is in fact much more
efficient: in particular, while the Hopper, et al. stegosystem requires
2 evaluations \emph{per bit} of a pseudorandom function, amounting to
a linear (in the key-size) number of applications of the underlying
PRG (in the standard construction for pseudorandom functions
of~\cite{DBLP:journals/jacm/GoldreichGM86}), in our stegosystem we
require \emph{per bit} a constant number of PRG applications.

\section{Definitions and Tools}
We say that a function $\mu: \mathbb{N} \to \mathbb{R}$ is \emph{negligible}
if for every positive polynomial $p(^.)$ there exists an $N$ such that for
all $n > N$, $\mu(n) < \frac{1}{p(n)}$. 

We let $\Sigma = \{ \sigma_1, \ldots, \sigma_s\}$ denote an alphabet and treat the \emph{channel}, which will be used for data transmission, as a family of random variables $\mathcal{C} = \{C_h\}_{h \in
  \Sigma^\ast}$; each $C_h$ is supported on $\Sigma$.  These channel distributions model a history-dependent notion of channel data: if $h_1, h_2, \ldots, h_\ell$ have been sent along the channel thus far, $C_{h_1,\ldots,h_\ell}$ determines the distribution of the next channel element.
  
\begin{definition}
A \emph{one-time stegosystem} consists of three probabilistic polynomial time 
algorithms
$$
S=(SK, SE, SD)
$$
 where:
\begin{itemize}
\item $SK$ is the \emph{key generation algorithm}; we write $SK\left(1^n, \log (1/ \epsilon_{sec}) \right)=k$. It takes as input, the security parameter $\epsilon_{sec}$ and the length of the message $n$ and produces a key $k$ of length $\kappa$. (We typically assume that $\kappa = \kappa(n)$ is a monotonically increasing function of $n$.)
\item $SE$ is the \emph{embedding procedure}, which can access the channel; $SE\left(1^n,k,m,h\right)=s \in \Sigma^*$. It takes as input the length of the message $n$, the key $k$, a
  message $m \in M_n \triangleq \{0,1\}^n$ to be embedded, and the history $h$ of previously
  drawn covertexts. The output is the stegotext $s \in \Sigma^*$.
\item $SD$ is the \emph{extraction procedure}; $SD\left(1^n,k,c \in \Sigma^\ast\right)= m \textrm{ or } \fail$. It takes as input $n$, $k$, and
  some $c \in \Sigma^\ast$. The output is a message $m$ or the token $\fail$.
\end{itemize}
\end{definition}

Recall that the \emph{min entropy} of a random variable $X$, taking values in a set $V$, is the quantity
$$
H_\infty(X) \triangleq \min_{v \in V} \left(-\log \Pr[X = v]\right)\,.
$$
We say that a channel $\mathcal{C}$ has min entropy $\delta$ if for all $h \in \Sigma^*$, $H_\infty(C_h) \geq \delta$.

\begin{definition}[Soundness] A stegosystem $(SK, SE, SD)$ is said to be \emph{$(s(\kappa),\delta)$-sound} provided that for all channels $\mathcal{C}$ of minimum entropy $\delta$,
\[ 
 \forall m\in M_n, \Pr[SD(1^n, k, SE(1^\kappa, k, m, h)) \neq m \mid k\gets SK(1^n, \log (1/ \epsilon_{sec})) ]
\leq \mbox{ s($\kappa$)}\,.
\]
\end{definition}

One-time stegosystem security is based on the indistinguishability between a transmission
that contains a steganographically embedded message and a transmission that
contains no embedded messages. 
An adversary $\mathcal{A}$ against a one-time stegosystem $S = (SK,
SE, SD)$ is a pair of algorithms $\mathcal{A}=(SA_1, SA_2)$, that
plays the following game, denoted $G^\mathcal{A}(1^n)$:
\begin{enumerate}
\item  A key $k$ is generated by $SK(1^n, \log (1/ \epsilon_{sec}))$.
  
\item Algorithm $SA_1$ receives as input the length of the message $n$ and outputs a triple $(m^\ast, s, h_{\sf c}) \in M_n \times
  \{0,1\}^\ast * \Sigma^*$, where $s$ is some additional information that will be passed to $SA_2$. $SA_1$ is provided access to $\mathcal{C}$ via an oracle $\mathcal{O}(h)$, which takes
  the history $h$ as input.

$\mathcal{O}(\cdot)$, on input $h$, returns to $SA_1$ an element $c$ selected according to $C_h$.
\item A bit $b$ is chosen uniformly at random. 
\begin{itemize}
\item If $b = 0$ let $c^\ast \gets SE(1^n,k,m^\ast, h)$, so $c^\ast$ is a stegotext. 
\item If $b = 1$ let $c^\ast = c_1 \circ \cdots \circ c_\lambda$, where $\circ$
denotes string concatenation and $c_i \stackrel{r}{\gets} C_{h \circ \sf c_1 \circ \cdots \circ c_{i-1}}$.
\end{itemize}
\item The input for $SA_2$ is $1^n$, $h_c$, $c^\ast$ and $s$. $SA_2$ outputs a bit $b'$. If $b' = b$ then we say that $(SA_1, SA_2)$ \emph{succeeded} and write $G^{\mc{A}}(1^n) = \text{success}$.
\end{enumerate}
The \emph{advantage} of the adversary $\mc{A}$ over a stegosystem $S$ is defined as: 
\[
 \mbox{{\bf Adv}}_S^\mathcal{A}(n) =  \left|\Pr\big[ G(1^n) = \text{success}  \big ] - \frac{1}{2} \right| \enspace.
 \]
 
 The probability includes the coin tosses of $\mathcal{A}$ and $SE$,
 as well as the coin tosses of $G(1^\kappa)$. The
 (information-theoretic) insecurity of the stegosystem is defined as
\[
\mbox{\bf{InSec}}_{S}(n) =  \max_{\mathcal{A} }\{\mbox{\bf{Adv}}_S^\mathcal{A}(n)\}\,,
\]
this maximum taken over all (time unbounded) adversaries $\mc{A}$.
\begin{definition}(Security)
We say that a stegosystem is \emph{$(t(n), \delta)$-secure} if for all channels with min entropy $\delta$ we have $\mbox{\bf{InSec}}_{S}(n) \leq t(n)$.
\end{definition}

\subsection{Error-correcting Codes}
\label{sec:codes}


Our steganographic construction requires an efficient family of codes that can recover from errors introduced by certain binary symmetric channels. In particular, we require an efficient version of the Shannon coding theorem~\cite{Shannon1, Shannon2}. For an element $x \in \{0,1\}^n$, we let $B_p(x)$ be the random variable equal to $x \oplus e$, where $e \in \{0,1\}^n$ is a random error vector defined by independently assigning each $e_i = 1$ with probability $p$. (Here $x \oplus e$ denotes the vector with $i$th coordinate equal to $x_i \oplus e_i$.) 

The classical coding theorem asserts that for every pair of real numbers $0 < R < C \leq 1$ and $n \in \mathbb{N}$, there is a binary code $A_n \subset \{0,1\}^n$, with $\log |A|/n \geq R$, so that for each $a \in A$, maximum-likelihood decoding recovers $a$ from $B_p(a)$ with probability $1 - e^{-\theta(n)}$, where 
\[H(p) = p \log p^{-1} + (1-p) \log (1-p)^{-1} = 1-C\,.
\]
The quantity $C$ (determined by $p$), is the \emph{capacity} of the binary symmetric channel induced by $B_p$; the quantity $R = \log |A|/n$ is the \emph{rate} of the code $A$. In this language, the coding theorem asserts that at rates lower than capacity, codes exist that correct random errors with exponentially decaying failure probability.

We formalize our requirements below:
\begin{definition}
  An error-correcting code is a pair of functions ${\textsf{E}} = (\Enc, \Dec)$, where $\Enc: \{0,1\}^n \rightarrow \{0,1\}^{\ell}$ is the \emph{encoding
  function} and $\Dec: \{0,1\}^{\ell} \rightarrow \{0,1\}^n$ the corresponding \emph{decoding function}. Specifically, we say that $\sf{E}$ is a \emph{$(n,\ell, p, \epsilon)$-code} if for all $m \in \{0,1\}^n$,
  \[
  \Pr[\Dec(\Enc(m) \oplus {e}) = m] \geq 1 - \epsilon
  \]
  where ${e} = (e_1, \ldots, e_{\ell})$ and each $e_i$ is independently
  distributed in $\{0,1\}$ so that $\Pr[e_i = 1] \leq p$. We say that $\textsf{E}$ is \emph{efficient} if both $Enc$ and $Dec$ are computable in polynomial time.
\end{definition}

\begin{proposition}
\label{proposition:forneycode}
Let $\tau = \tau(n)$ lie in the interval $(0,1/4)$, $p = 1/2 - \tau$, and $R^{\prime} = 1 - H(p)$. Let $n \geq 16$ be a message length for which ${\left(104 \log\left(\log n\right)\right)^{3}}/{\log n} \leq \tau^2\,.$
 Then there is an efficient family of 
$(n, \ell(n), p, \epsilon(n))$-error-correcting codes $\textsf{E}_n$ for which 
\[
\epsilon(n) \leq e^{- 4n/\log n} \quad \text{and} \quad \ell(n) \leq (1+57/\sqrt[3]{\tau^{2}\log n})^{2}n/R^{\prime}\,.
\]
\end{proposition}

\begin{proof}
This is a consequence of Forney's~\cite{forney} efficient realizations of the Shannon coding theorem~\cite{Shannon1, Shannon2}; we work out the technical details in the full version of the paper.

\end{proof}

We refer to~\cite{vanLint,Gallager:IT65} for detailed discussions of error-correcting codes over binary symmetric channels.

\subsection{Function Families and Almost $t$-wise Independence}
\label{Subsec:t-wise}

We will employ the notion of (almost) $t$-wise independent function
families (cf.~\cite{DBLP:journals/rsa/AlonGHP92}, \cite{DBLP:journals/siamcomp/NaorN93}).
\begin{definition}
\label{def:t-away}
A family $\mathcal{F}$ of Boolean functions on $\{0,1\}^n$ is said
to be \emph{$\epsilon$-away from $t$-wise independent} or
\emph{$(n,t,\epsilon)$-independent} if for any $t$ distinct domain elements
$q_1, q_2, \dots,q_t$ we have
\begin{equation}\label{t-wise-1}
  \sum \limits_{\alpha \in \{0,1\}^t} \left| \Pr_{f}[f_k(q_1)f_k(q_2) \cdots f_k(q_t)=\alpha]-\frac{1}{2^t}\right| \leq \epsilon\,,
\end{equation}
where $f$ chosen uniformly from $\mathcal{F}$.
\end{definition}
The above is equivalent to the following formulation
quantified over all computationally unbounded adversaries
$\mathcal{A}$:
\begin{equation}\label{t-wise-2}
\Big |\Pr_{f \stackrel{r}\gets \mathcal{F}}[G^{\mathcal{A}^{f[t]}}(1^\kappa) = 1] - \Pr_{f \stackrel{r}\gets  \mathcal{R}}[G^{\mathcal{A}^{f[t]}}(1^\kappa)=1]\Big| \leq \epsilon\,,
\end{equation}
\noindent
where $\mathcal{R}$ is the collection of \emph{all} functions from
$\{0,1\}^n$ to $\{0,1\}$ and $\mathcal{A}^{f[t]}$ is an unbounded
adversary that is allowed to determine up to $t$ queries to the
function $f$ before he outputs his bit.
  \begin{lemma}
    \label{lemma:1-2}
    $\mathcal{F}_\kappa$ is $\epsilon'$-away from $t$-wise independence according to
    equation (\ref{t-wise-1}) if and only if $\mathcal{F}_\kappa$ is
    $\epsilon'$-away from $t$-wise independence according to equation
    (\ref{t-wise-2}) above.
  \end{lemma}
We employ the construction of almost $t$-wise independent sample
spaces given by \cite{DBLP:journals/siamcomp/NaorN93},  \cite{DBLP:journals/rsa/AlonGHP92}.
\begin{theorem}[\cite{DBLP:journals/siamcomp/NaorN93},  \cite{DBLP:journals/rsa/AlonGHP92}]
  \label{thm:independence}
 There exist families of Boolean functions $\mathcal F^n_{t,\epsilon}$ on $\{0,1\}^n$ that are
  $\epsilon$-away from $t$-wise independent, are indexed by keys of length $(2 + o(1))(\log\log n + \frac{t}{2} + \log{1/\epsilon})$, and are computable in polynomial time. 
\end{theorem}
\subsection{Rejection Sampling}
\label{subsec:rejsam}

A common method used in steganography employing a channel distribution
is that of \emph{rejection sampling} (cf.
\cite{DBLP:conf/ih/Cachin98,DBLP:conf/crypto/HopperLA02}).  Assuming
that one wishes to transmit a single bit $m$ and employs a random
function $f: \{0,1\}^d \times \Sigma \to \{0,1\}$ that is secret from the adversary, one performs the
following ``rejection sampling'' process:
\begin{center}
\begin{tabular}{|l|}
\hline
\texttt{rejsam$^{f}_h(m)$}\\
\hline
\hspace{1cm} $c \stackrel{r}\gets C_h $\\
\hspace{1cm} if $f(c) \neq m$\\
\hspace{1cm} then $c \stackrel{r}\gets C_h$\\
\texttt{Output:} $c$\\
\hline
\end{tabular}
\end{center}
Here, as above, $\Sigma$ denotes the output alphabet of the channel, $h$ denotes the
history of the channel data at the start of the process, and
$C_h$ denotes the distribution on $\Sigma$ given by the channel after
history $h$.  The receiver (also privy to the function $f$)
applies the function to the received message $c \in \Sigma$ and recovers $m$
with probability greater than $1/2$.  The sender and the receiver may
employ a joint state denoted by $i$ in the above process (e.g., a
counter), that need not be secret from the adversary.
Note that the above process performs only two draws from the channel
with the \emph{same} history (more draws could, in principle, be
performed).  These draws are assumed to be independent.  One
basic property of rejection sampling that we use is:
\begin{lemma}
  \label{lem:rej-samp}
  If $f$ is drawn uniformly at random from the collection of all
  functions $\mathcal{R} = \{ f : \Sigma \to
  \{0,1\}\;\}$ and $\mathcal{C}$ has min entropy $\delta$, 
  then
  \[
  \Pr_{f \gets \mathcal{R}}[f({\texttt{rejsam}}_h^{f}(m)) = m] \geq \frac{1}{2}+\tau \enspace,
  \]
  where $\tau = \frac{1}{4}\left(1-\frac{1}{2^\delta}\right)$. 
\end{lemma}
\begin{proof}
%
Define the event $E$ to be 
\[
E = [f({\sf c}_1  ) =  m]   \lor  [f({\sf c}_1) \neq  m  \land  f({\sf c}_2 ) =  m ] \,;
\]
thus $E$ is the event that rejection sampling is successful for $m$.
Here ${\sf c}_1, {\sf c}_2$ are two independent random variables distributed according
to the channel distribution $C_h$ and $h$ is determined by the history of channel usage. 
%
Recalling that $\Sigma = \{\sigma_1, \dots, \sigma_s\}$ is the support of the channel distribution $C_h$, let $p_i = \Pr[C_h = \sigma_i]$ denote the probability that $\sigma_i$ occurs. As $f$ is chosen uniformly at random, 
$$
\Pr[f({\sf c}_1  ) =  m] = 1/2\,.
$$
Then $\Pr[E] = 1/2 + \Pr[A]$, where $A$ is the event that $f({\sf c}_1) \neq  m  \land  f({\sf c}_2 ) =  m $. To bound $\Pr[A]$, let $D$ denote the event that $\textsf{c}_1 \neq \textsf{c}_2$.  Observe that conditioned on $D$, $A$ occurs with probability exactly $1/4$; on the other hand, $A$ cannot occur simultaneously with $\overline{D}$. Thus
\[
\mbox{Pr}[E] = \frac{1}{2} + \Pr[A\mid D] \cdot \Pr[D] + \Pr[A \mid \overline{D}]\cdot \Pr[\overline{D}]  = \frac{1}{2} + \frac{1}{4} \Pr[D]\enspace.
\]
To bound $\Pr[D]$, note that 
\[
\Pr[\bar{D}] = \sum_i p_i^2 \leq \max_i p_i \sum_i p_i = \max_i p_i
\]
and hence that $\Pr[D] \geq 1 - \max_i p_i$. Considering that $H_\infty(C) \geq \delta$, we have $\max_i p_i \leq
\frac{1}{2^\delta}$ and the success probability is
\[
\Pr[E]   \geq  \frac{1}{2} +\frac{1}{4}\cdot (1- p_i)  \geq  \frac{1}{2}  +\frac{1}{4}\left(1-\frac{1}{2^\delta}\right) = \frac{1}{2}+\tau \,,
\]
where $\tau = \frac{1}{4}\left(1-\frac{1}{2^\delta}\right)$.
\end{proof}

\section{The construction}

In this section we outline our construction of a one-time stegosystem
as an interaction between Alice (the sender) and Bob (the receiver).
Alice and Bob wish to communicate over a channel with distribution
$\mathcal{C}$. We assume that $\mathcal{C}$ has min entropy $\delta$, so that $\forall h \in \Sigma^\ast$, $H_\infty(C_h) \geq \delta$.
As above, let $\tau = \frac{1}{4}\left(1-\frac{1}{2^\delta}\right)$. For simplicity, we assume that the support of
$\mathcal{C}_h$ is of size $|\Sigma|=2^b$.
\subsection{A one-time stegosystem} 
\label{sec:stegosystem}
Fix an alphabet $\Sigma$ for the channel and choose a message length $n$ and security parameter $\epsilon_{\mathcal F}$.
Alice and Bob agree on the following:
\begin{sloppypar}
  \begin{description}
  \item[An error-correcting code.] Let ${\sf{E}} = (Enc, Dec)$ be an efficient $(n, \lambda, \frac{1}{2}-\tau, \epsilon_{\text{enc}})$-error-correcting code;
  \item[A pseudorandom function family.] Let $\mathcal{F}$ be a
    function family that is $(\log \lambda + \log |\Sigma|, 2\lambda, \epsilon_{\mathcal
      F})$-independent. We treat elements of $\mc{F}$ as Boolean functions on $\{1, \ldots, \lambda\} \times \Sigma$ and, for such a function $f$ we let $f_i: \Sigma \rightarrow \{0,1\}$ denote the function $f_i(\sigma) = f(i,\sigma)$.
  \end{description}
\end{sloppypar}
\noindent
We will analyze the stegosystem below in terms of arbitrary parameters $\lambda$, $\epsilon_{\mathcal{F}}$, and $\epsilon_\text{enc}$, relegating discussion of how these parameters determine the overall efficiency of the system to Section~\ref{sec:parameters}.
 
 Key generation consists of selecting an element $f \in \mc{F}$.
Alice and Bob then communicate using the algorithms $SE$ for
embedding and $SD$ for extracting as described in
Figure~\ref{fig:protocol}.
\begin{figure}[ht]
  \begin{center}
    \begin{tabular}{@{}|l|l|}\hline \label{protocol} 
      \ttfamily PROCEDURE $SE$: & {\ttfamily PROCEDURE} $SD$:  \\ 
      \hline
      \ttfamily Input: Key $k$, hidden text $m'$,& \ttfamily Input: Key $k$, stegotext $\sf{c_{stego}}$ \\
      \hspace{1.3cm}\ttfamily history $h$&\\ 
      \ttfamily let $m = Enc(m')$ & \\
      \ttfamily parse $m$ as  $m = m_{1} m_{2} \dots m_{\lambda}$ & \ttfamily parse $c_{stego}$ as  $c= c_1 c_2 \ldots c_{\lambda}$ \\
      \ttfamily for $i=1$ to $\lambda$ \{ &  \ttfamily for $i=1$ to $\lambda$ \{\\
      \hspace{1 cm} $c_i=$\texttt{rejsam}$^{f_i}_h(m_i)$& \hspace{1 cm} \ttfamily set $\bar{m}_i = f_i(c_i)$ \\
      \hspace{1 cm} \ttfamily set $h \gets h \circ c_i$ &  \hspace{1 cm} \ttfamily let $\bar{m}=\bar{m}_{1}\bar{m}_{2}\ldots\bar{m}_{\lambda }$ \\
      \} &  \ttfamily \}  \\
      \ttfamily Output: $c_{stego}=c_1 c_2\ldots  c_{\lambda} \in \Sigma^\lambda$&\ttfamily Output: $Dec(\bar{m})$  \\
      \hline
    \end{tabular}\\
  \end{center}
\caption{Encryption and Decryption algorithms for the one-time
  stegosystem of \ref{protocol}.}
\label{fig:protocol}
\end{figure}
In $SE$, after applying the error-correcting code ${\sf{E}}$, we use
\texttt{rejsam$^{f_i}_h(m_i)$} to obtain an element $c_i$ of the channel for each bit $m_i$ of the message. The resulting stegotext $c_1\ldots c_\lambda$ is denoted $c_{stego}$.  In $SD$ the received stegotext is parsed block by
block by evaluating the key function $f_i$ at $c_i$; this results in a message bit. After performing this for each received block, a message of
size $\lambda$ is received, which is subjected to decoding via $\Dec$.  Note that we sample at
most twice from the channel for each bit we wish to send.  The
error-correcting code is needed to recover from the errors introduced
by this process.  The detailed security and correctness analysis
follow in the next two sections.

\subsection{Correctness}
\label{sec:soundness}
We focus on the mapping between $\{0,1\}^\lambda$ and $\Sigma^\lambda$ determined by the $SE$ procedure of the one-time stegosystem.  In particular, for an initial history $h$ and a key
function $f: \{1, \ldots, \lambda\} \times \Sigma \to \{0,1\}$, 
\begin{figure}[ht]
\begin{center}
  \texttt{ $\mathrm{P}^f_h:\{0,1\}^\lambda \to \Sigma^\lambda$
    \begin{tabular}{|p{10cm}}
      input: $h$, $m = m_1\ldots m_\lambda \in \{0,1\}^\lambda$\\
      \hspace{1cm}\texttt{for $i = 1$ to $\lambda$}\\
      \hspace{1cm}$c_i =$\texttt{rejsam}$^{f_i}_h (m_i)$\\
      \hspace{1cm}$h \gets h \circ c_i$\\
      output: $c = c_1\ldots c_\lambda \in \Sigma^\lambda$ \\
    \end{tabular}
}
\end{center}
\caption{The procedure $P^f_h$.}
\label{fig:P}
\end{figure}
recall that the covertext of the message $m$ is given by the procedure $P^{f}(m) = P^{f}_h(m)$, described in Figure~\ref{fig:P}; here $h$ is the initial history. We remark now that the procedure defining
$P^f$ samples $f$ at no more than $2\lambda$ points and that the family
$\mathcal F$ used in $SE$ is $\epsilon_{\mathcal F}$-away from $2\lambda$-wise
independent.  For a string $c = c_1 \ldots c_\lambda \in \Sigma^\lambda$ and a function $f$,
let $R^f(c) = (f_1(c_1), \ldots, f_\lambda(c_\lambda)) \in \{0,1\}^\lambda$. If $f$ were chosen uniformly among \emph{all} Boolean functions on $\{1, \ldots, \lambda\} \times \Sigma$ then we could conclude from
Lemma~\ref{lem:rej-samp} above that each bit is independently recovered by
this process with probability at least $\frac{1}{2}+\tau$. As \textsf{E} is an $(n,
\lambda, \frac{1}{2}-\tau, \epsilon_{\text{enc}})$-error-correcting code, this would imply that
\[
\Pr_{f \gets \mathcal{R}} [ R^f(P^f_h(m)) = m ] \geq 1 - \epsilon_{\text{enc}}\enspace.
\]
This is a restatement of the correctness analysis of Hopper, et al
\cite{DBLP:conf/crypto/HopperLA02}. Recalling that the procedure
defining $R^f(P^f_h(\cdot))$ involves no more than $2\lambda$ samples
of $f$, condition~\eqref{t-wise-2} following
Definition~\ref{def:t-away} implies that
\begin{equation}
  \label{eq:correct}
  \Pr_{f \gets \mathcal{F}} [ R^f(P^f_h(m)) = m ] \geq 1 - \epsilon_{\text{enc}} -
  \epsilon_{\mathcal F}
\end{equation}
so long as $\mathcal{F}$ is $(\log \lambda + \log |\Sigma|, 2 \lambda, \epsilon_{\mathcal F})$-independent.
(We remark that as described above, the procedure $P^f_h$ depends on
the behavior of channel; note, however, that if there were a sequence
of channel distributions which violated~\eqref{eq:correct} then there
would be a fixed sequence of channel responses, and thus a
deterministic process $P^f$, which also violated~\eqref{eq:correct}.)
To summarize
\begin{lemma}
\label{lem:soundness}
  With $SE$ and $SD$ described as above, the probability that a
  message $m$ is recovered from the stegosystem is at least $1 -
  \epsilon_{\text{enc}} - \epsilon_{\mathcal F}$. 
\end{lemma}

\subsection{Security}
\label{sec:security}

In this section we argue about the security of our one-time
stegosystem.  First we will observe that the output of the rejection
sampling function ${\tt rejsam}^{f}_h$, with a truly random function $f$, is
indistinguishable from the channel distribution $\mathcal{C}_h$. (This is a folklore result implicit in previous work.)  We then show that if $f$ is selected from a family that is
$\epsilon_{\mathcal{F}}$-away from $2\lambda$-wise independent, the
advantage of an adversary $\mathcal{A}$ to distinguish between the
output of the protocol and $\mathcal{C}_h$ is bounded above by $\epsilon_{\mathcal{F}}$.
Let $\mathcal{R} =\{f : \Sigma \to \{0,1\} \}$. First we characterize the
probability distribution of the rejection sampling function:
\begin{proposition}
  \label{proposition:distribution}
  The function ${\tt rejsam}^{f}_h(m)$ is a random variable with
  probability distribution expressed by the following function: Let $c
  \in \Sigma$ and $m\in \{0,1\}$. Let $\exclude_{f}(m) = \Pr_{c' \gets \mathcal{C}_h} [f(c') \neq m]$ and $p_c = \Pr_{c' \gets \mathcal{C}_h}[c' = c]$. Then
  \[
  \Pr[ {\tt rejsam}^{f}_h(m) = c] =
  \begin{cases}
   p_c\cdot(1+  \exclude_{f}(m)) & \text{if}\;f(c) = m\enspace, \\
   p_c \cdot \exclude_{f}(m) & \text{if}\;f(c) \neq m\enspace.\\
  \end{cases}
  \]
\end{proposition}
  \begin{proof}
    Let $c_1$ and $c_2$ be the two (independent) samples drawn from $\mc{C}_h$
    during rejection sampling. (For simplicity, we treat the process as
    having drawn two samples even in the case where it succeeds on the
    first draw.)     
    Note, now, that in the case where $f(c) \neq m$, the value $c$ is
    the result of the rejection sampling process precisely when
    $f(c_1) \neq m$ and $c_2 = c$; as these samples are independent,
    this occurs with probability $\exclude_{f}(m) \cdot p_c$.
    
    In the case where $f(c) = m$, however, we observe $c$ whenever
    $c_1 = c$ or $f(c_1) \neq m$ and $c_2 = c$. As these events are
    disjoint, their union occurs with probability $p_c \cdot (\exclude_{f}(m) +
    1)$, as desired.
  \end{proof}
\begin{lemma}\label{lemma:indist-channel}
  For any $h \in \Sigma^*, m \in \{0,1\}$, the random variable ${\tt rejsam}^{f}_h(m)$ is
  perfectly indistinguishable from the channel distribution
  $C_h$ when $f$ is drawn uniformly at random from the space
  of $\mc{R}$.
\end{lemma}
  \begin{proof}
    Let $f$ be a random function, as described in the statement of the
    lemma. Fixing the elements $c$, and $m$, we condition on the
    event $E_{\neq}$, that $f(c) \neq m$. In light of
    Proposition~\ref{proposition:distribution}, for any $f$ drawn under
    this conditioning we shall have that $\Pr[ {\tt rejsam}^{f}_h(m) =
    c]$ is equal to
    $$
    \Pr_{c' \gets \mathcal{C}_h}[c' = c] \cdot \exclude_{f}(m) = p_c \cdot \exclude_{f}(m)\enspace,
    $$
    where we have written $\exclude_{f}(m) = \Pr_{c' \gets \mathcal{C}_h}
    [f(c') \neq m]$ and $p_c = \Pr_{c' \gets \mathcal{C}_h}[c' = c]$. Conditioned on $E_{\neq}$, then, the probability
    of observing $c$ is
    $$
    \mathbf{E}_f [p_c \cdot \exclude_{f}(m) \mid E_{\neq}] = p_c \left( p_c + \frac{1}{2}(1 - p_c)\right)\enspace.
    $$
    Letting $E_=$ be the event that $f(i,c) = m$, we similarly compute
    $$
    \mathbf{E}_f [p_c \cdot \exclude_{f}(m) \mid E_{=}] = p_c \left(1 +  \frac{1}{2}(1 - p_c)\right)\enspace.
    $$
    As $\Pr[E_=] = \Pr[E_{\neq}] = 1/2$, we conclude that the
    probability of observing $c$ is exactly
    $$
    \frac{1}{2}\left(p_c \left( p_c + \frac{1 - p_c}{2}\right) + p_c
      \left(1 + \frac{1 - p_c}{2}\right)\right) = p_c\enspace,
    $$
    as desired.
  \end{proof}
\noindent

The following corollary follows immediately from the lemma above.

\begin{corollary}
\label{cor:channeldist}
  For any $h \in \Sigma^*, m \in \{0,1\}^{\lambda}$, the random variable $P^f_h$ is
  perfectly indistinguishable from the channel distribution
  $C^{\lambda}_h$ when $f$ is drawn uniformly at random from the space
  of all Boolean functions on $\{1, \ldots, \lambda\} \times \Sigma$.
\end{corollary}

Having established the behavior of the rejection sampling function
when a truly random function is used, we proceed to examine the
behavior of rejection sampling in our setting where the function is
drawn from a function family that is $\epsilon_{\mathcal{F}}$-away from $2\lambda$-wise
independence. In particular we will show that the insecurity of the
defined stegosystem is characterized as follows:
\begin{lemma}
\label{lem:security}
The insecurity of the stegosystem $S$ of Section~\ref{sec:stegosystem}
is bound by $\epsilon_{\mathcal{F}}$, i.e., $\mathbf{InSec}_{S}(n) \leq \epsilon_{\mathcal{F}}$, where $\epsilon_{\mathcal{F}}$ is the
bias of the almost $2\lambda$-wise independent function family employed;
recall that $\lambda = \ell(n)$ is the stretching of the input incurred due to
the error-correcting code.
\end{lemma}
\noindent

\begin{proof}
Let us play the following game $G(1^\kappa)$ with the adversary $\mathcal{A}$. 

In each round we either select $G_1^\mathcal{A}$ or $G_2^\mathcal{A}$: 
\paragraph{}
\begin{tabular}{|l|l|}
\hline
$G_1^\mathcal{A}(1^\kappa$)&\\
\hline
1.& $k \gets \{0,1\}^\kappa$\\
2.& $(m^\ast,s)\gets SA^{\mathcal{O}(h)}_1(1^\kappa,h)$, $m^\ast \in \{0,1\}^n $ \\
3.& $b \stackrel{r}{\gets} \{0,1\}$\\
4.&$c^\ast$ = $\left \{ \begin{array}{lll}
c_0, c_1, \dots c_{\lambda-1} & c_i = \mbox{\texttt{rejsam}}^{f_k, i}_h(m_i), h = h \circ c & \hspace{0.5 cm}  \textrm{ if }b=0\\
\mbox{from the channel} && \hspace{0.5 cm} \textrm{ if } b=1
\end{array} \right .$\\
5. & $b^\ast \gets SA_2(c^\ast,s)$\\
6. & if $ b = b^\ast$ then success\\
\hline
\end{tabular}
\paragraph{}
\begin{tabular}{|l|l|}
\hline
$G_2^\mathcal{A}(1^\kappa)$\\
\hline
1.& $f \gets \mathcal{R}$ \\
2.& $(m^\ast,s)\gets SA^{\mathcal{O}(h)}_1(1^\kappa,h)$, $m^\ast \in \{0,1\}^n $ \\
3.& $b \stackrel{r}{\gets} \{0,1\}$\\
4.&$c^\ast$ = $\left \{ \begin{array}{lll}
c_0, c_1, \dots c_{\lambda-1} & c_i = \mbox{\texttt{rejsam}}^{f, i}_h(m_i), h = h \circ c &\hspace{0.5 cm} \textrm{ if } b=0\\
\mbox{from the channel} && \hspace{0.5 cm}  \textrm{ if } b=1
\end{array} \right .$\\
5.& $b^\ast \gets SA_2(c^\ast,s)$\\
6.& if $ b = b^\ast$ then success\\
\hline
\end{tabular}
\paragraph{}
\begin{eqnarray*}
\textbf{Adv}_{S}^{\mathcal{A}}(G(1^\kappa))& = & \Big | \Pr [\mathcal{A}^{{\mathcal{O}(h)},c^\ast \gets SE(k,^.,^.,^.)} = 1]
- \Pr[A^{{\mathcal{O}(h)}, c^\ast \gets \mathcal{C}_h  } = 1] \Big | \\
&=& \Pr_{f \gets \mathcal{F}_\kappa}[G(1^\kappa)=1]- \Pr_{f \gets \mathcal{R}}[G(1^\kappa)=1] \leq \epsilon_{\mathcal{F}} 
\end{eqnarray*}
and the lemma follows by the definition of insecurity.
\end{proof}

\subsection{Putting it all together}
\label{sec:parameters}

The objective of this section is to integrate the results of the previous sections of the paper into one unifying theorem. As our system is built over two-sample rejection sampling, a process that faithfully transmits each bit with probability $1/2 + \tau$, we cannot hope to achieve rate exceeding 
$$
R' = 1 - H(1/2 + \tau) = 1 - H(1/4 + 2^{-\delta}/4)\,.
$$
Indeed, as described in the theorem below, the system asymptotically converges to the rate of this underlying rejection sampling channel. (We remark that with sufficiently large channel entropy, one can draw more samples during rejection sampling without interfering with security; this can control the noise introduced by rejection sampling.)

\begin{theorem}

For $\delta = \Omega(\sqrt{(\log\log n)^{3}/\log n})$
the stegosystem $S$ uses private keys $k$ of length no more than 
$$
(2 + o(1))\left[\lambda(n) + \log 1/\epsilon_\mathcal{F} + \log\log\log |\Sigma|\right]
$$
and is both $(\epsilon_{enc} + \epsilon_\mathcal{F},\delta)$-sound and $(\epsilon_\mathcal{F},\delta)$-secure. The length of the stegotext $\lambda(n)$ is

$$\lambda(n) \leq \left(1+\frac{1}{\log(\log n)}\right)^{2}\frac{n}{R^{\prime}}\,,$$
where $\epsilon_{enc} \leq e^{-4n/\log n}$ and $R^{\prime} = 1-H\left(1/4+2^{-\delta}/4\right)$.
\end{theorem}

\begin{proof}
Let $\Sigma = \{ \sigma_1, \ldots, \sigma_s\}$ denote an alphabet and define the \emph{channel} as a family of random variables $\mathcal{C} = \{C_h\}_{h \in \Sigma^\ast}$; each $C_h$ supported on $\Sigma$. Also, the channel $\mathcal{C}$ has min entropy $\delta$, so that $\forall h \in \Sigma^\ast$, $H_\infty(C_h) \geq \delta$. Fix an alphabet $\Sigma$ for the channel and choose a message length $n \geq 16$ such that  
$$
-\log\left(1-4\sqrt{{104 \log\left(\log n\right)^{3}}/{\log n}}\right) \leq \delta\,.
$$
Under the assumption that the channel $\mathcal{C}$ has min entropy $\delta$, the binary symmetric channel induced by the rejection sampling process of Lemma~\ref{lem:rej-samp} has transition probability no more than $1/4(1 + 2^{-\delta})$.
We have an efficient $(n, \lambda, \frac{1}{4}(1 + 2^{-\delta}), \epsilon_{\text{enc}})$ error-correcting code as discussed in Section~\ref{sec:codes} that encodes messages of length $n$ as codewords of length

$$\lambda(n) = \left(1+\frac{57}{\sqrt[3]{\tau^{2}\log n}}\right)^{2}\frac{n}{R^{\prime}} \leq \left(1+\frac{1}{\log(\log n)}\right)^{2}\frac{n}{R^{\prime}}\textrm{ bits}$$
\[
(2 + o(1))\left[ \lambda(n) +\log 1/\epsilon_\mathcal{F} + \log\log\log |\Sigma| \right] \\
\]
random bits; these serve as the key for the stegosystem. 
In light of the conclusions of Lemma~\ref{lem:security} and Lemma~\ref{lem:soundness}, this system achieves the $(\epsilon_{\text{enc}} + \epsilon_\mathcal{F},\delta)$-soundness and $(\epsilon_\mathcal{F},\delta)$-security.
\end{proof}

For concreteness, we record two corollaries:
\begin{corollary}
\label{cor:conc}

There exists a function $\delta(n) = o(1)$ 
so that the stego system $S$, using private keys $k$ of length no more than
$$
O\left(n + \log|\Sigma| +\log 1/\epsilon_\mathcal{F}\right)\,,
$$
is both $(e^{-4n/\log n} + \epsilon_\mathcal{F},\delta)$-sound and $(\epsilon_\mathcal{F},\delta)$-secure. Here, the length of the stegotext is

$$\lambda(n) = \left(1+o(1)\right)\frac{n}{R^{\prime}}$$
where $R^{\prime} = 1-H\left(1/4(1 +2^{-\delta})\right)$.

\end{corollary}

\begin{corollary}
\label{corr:cons}
For any constant $\delta$, the stegosystem $S$ uses private keys of length $O(n + \log \Sigma + \log \epsilon_{\mathcal{F}})$ and transmits no more than $O(n)$ symbols. 
\end{corollary}

\section{A provably secure stegosystem for longer messages}
In this section we show how to apply the ``one-time'' stegosystem 
of Section~\ref{sec:stegosystem} together with  a pseudorandom number
generator so that longer messages can be transmitted.
\begin{definition}
Let $U_l$ denote the uniform distribution over $\{0,1\}^l$.
A polynomial time deterministic program $G$ is a pseudorandom
generator (PRG) if the following conditions are satisfied:
\begin{description}
\item[Variable output] For all seeds $x \in \{0,1\}^\ast$ and $y \in
\mathbb{N}$, $|G(x,1^y)|=y$ and, furthermore, $G(x,1^y)$ is a prefix
of $G(x,1^{y+1})$.
\item[Pseudorandomness] For every polynomial $p$ the
set of random variables $\{G(U_l, 1^{p(l)} )\}_{l \in {\rm N}}$ is
computationally indistinguishable from the uniform distribution
$U_{p(l)}$.
\end{description}
\end{definition}
Note that there is a procedure $G'$ that if $z = G(x, 1^y)$ it holds
that $G(x, 1^{y+y'}) = G'(x, z, 1^{y'})$ (i.e., if one maintains $z$,
one can extract the $y'$ bits that follow the first $y$ bits without
starting from the beginning).  For a PRG $G$, if $A$ is some
statistical test, then we define the advantage of $A$ over the PRNG as
follows:
\[\mbox{\bf Adv}_{G}^A(l) = \left| \Pr_{\hat{l} \gets G(U_l, 1^{p(l)})}
[A(\hat{l}) = 1] - \Pr_{\hat{l} \gets U_{p(l)}}[A(\hat{l}) = 1]\right|\] 
The insecurity of the PRNG $G$ is 
then defined 
$$
\mbox{\bf{InSec}}^{PRG}_{G}(l) =
\mbox{max}_{A}\{\mbox{\bf{Adv}}_{G}^A (l)\}\enspace.
$$
Note that typically in PRGs there is a procedure $G'$ as well as the
process $G(x, 1^y)$ produces some auxiliary data ${\sf aux}_y$ of
small length so that the rightmost $y'$ bits of $G(x, 1^{y+y'})$ may
be sampled directly as $G'(x, 1^{y'}, {\sf aux}_y)$.
Consider now the following stegosystem $S' = (SE', SD')$ that can be
used for arbitrary many and long messages and employs a PRG $G$ and
the one-time stegosystem $(SK, SE, SD)$ of Section~\ref{sec:stegosystem}.
The two players Alice and Bob, share a key of length $l$ denoted by
$x$. They also maintain a state $N$ that holds the number of bits that
have been transmitted already as well the auxiliary information ${\sf
  aux}_N$ (initially empty).  The function $SE'$ is given input $N,
{\sf aux}_N, x, m \in\{0,1\}^{n}$ where $m$ is the message to be
transmitted. $SE'$ in turn employs the PRG $G$ to extract a number of
bits $\kappa$ as follows $k = G'(x, 1^{\kappa}, {\sf aux}_N)$. The length $\kappa$
is selected to match the number of key bits that are required to
transmit the message $m$ using the one-time stegosystem of
section~\ref{sec:stegosystem}.  Once the key $k$ is produced by the
PRG the procedure $SE'$ invokes the one-time stegosystem on input $k,
m, h$.  After the transmission is completed the history $h$, the count
$N$, as well as the auxiliary PRG information ${\sf aux}_N$ are
updated accordingly.  The function $SD'$ is defined in a
straightforward way based on $SD$.
\begin{theorem}
  The stegosystem $S' = (SE' , SD')$ is provably secure in the model
  of \cite{DBLP:conf/crypto/HopperLA02} (universally
  steganographically secret against chosen hiddentext attacks); in
  particular 
$$\mathbf{InSec}^{SS}_{S'}(t, q, l) \leq
  \mathbf{InSec}^{PRG}(t+ \gamma(\ell(l)), \ell( l ) + {\sf polylog}(l) )
$$
  (where $t$ is the time required by the adversary, $q$ is the number
  of chosen hiddentext queries it makes, $l$ is the total number of
  bits across all queries and $\gamma(v)$ is the time required to simulate
  the $SE'$ oracle for $v$ bits).
\end{theorem} 
\subsection{Performance Comparison of the Stegosystem $S'$ and the
   Hopper, Langford, von Ahn System }

The system of Hopper, et al.~\cite{DBLP:conf/crypto/HopperLA02} concerns a situation where the min entropy of all $\mathcal{C}_h$ is at least 1 bit. In this case, we may select an $(n,\lambda,3/8,\epsilon_{\text{enc}})$-error-correcting code
\textsf{E}. Then the system of Hopper, et
al. correctly decodes a given
message with probability at least $1-\epsilon_{\text{enc}}$ and makes no
more than $2\lambda$ calls to a pseudorandom function family. Were one to
use the pseudorandom function family of Goldreich, Goldwasser, and
Micali~\cite{DBLP:journals/jacm/GoldreichGM86}, then this involves
production of $\Theta(\lambda \cdot k \cdot (\log(|\Sigma|) + \log \lambda))$ pseudorandom bits,
where $k$ is the security parameter of the pseudorandom function
family. Of course, the security of the system depends on the security
of the underlying pseudorandom generator.  
On the other hand, with the same error-correcting code, the
steganographic system described above utilizes $O\left[ \log\log\log |\Sigma| +\lambda + \log 1/\epsilon_\mathcal{F}\right]$ pseudorandom bits, correctly decodes a
given message with probability $1-(\epsilon_{\text{enc}} + \epsilon_{\mathcal{F}})$, and
possesses insecurity no more than $\epsilon_{\mathcal{F}}$. In order to
compare the two schemes, note that by selecting $\epsilon_{\mathcal{F}} =
2^{-k}$, both the decoding error and the security of the two systems
differ by at most $2^{-k}$, a negligible function in terms of the
security parameter $k$. (Note also that pseudorandom functions
utilized in the above scheme have security no better than $2^{-k}$
with security parameter $k$.)  In this case, the number of pseudorandom bits used by our system, 
$$
(2 + o(1))\bigl[ \lambda(n) + \log 1/\epsilon_\mathcal{F} + \log\log\log |\Sigma|\bigr]\,,
$$
is a dramatic improvement over the
$\Theta(\lambda k \log (|\Sigma|\lambda))$ bits of the scheme above.

\bibliographystyle{plain}
\bibliography{ih-jcrypto}

\end{document}